%% file: main.tex
\documentclass{article}

\usepackage[preprint]{neurips_2026}

\usepackage[utf8]{inputenc} %
\usepackage[T1]{fontenc}    %
\usepackage{hyperref}       %
\usepackage{url}            %
\usepackage{booktabs}       %
\usepackage{amsfonts}       %
\usepackage{nicefrac}       %
\usepackage{microtype}      %
\usepackage{xcolor}         %
\usepackage{amsthm}
\usepackage{amsmath}
\usepackage{amssymb}
\usepackage{graphicx}
\usepackage{subcaption}

\newcommand{\M}{\mathcal{M}}
\newcommand{\eps}{\varepsilon}
\newcommand{\veps}{\boldsymbol{\varepsilon}}
\renewcommand{\Pr}{\mathbb{P}}
\newcommand{\mse}{\operatorname{MSE}}
\newcommand{\thr}{\mathrm{thr}}

\newtheorem{definition}{Definition}
\newtheorem{theorem}{Theorem}

\newtheorem{lemma}{Lemma}
\newtheorem{claim}{Claim}

\newtheorem{remark}{Remark}
\newtheorem{example}{Example}
\newcommand{\Lap}{\operatorname{Lap}}
\usepackage{dsfont}
\newcommand{\uno}{\mathds{1}}
\newcommand{\mup}{\mu_{\mathcal{P}}}
\title{Limits of Personalizing Differential Privacy Budgets}

\author{%
  Edwige Cyffers \\
  CNRS, LAMSADE, Dauphine-PSL\\
\texttt{edwige.cyffers@cnrs.fr} \\
\And
  Juba Ziani\\
  Georgia Institute of Technology.\\
  \texttt{jziani3@gatech.edu} \\
}

\begin{document}

\maketitle

\begin{abstract}
A key technical difficulty in differential privacy is selecting a privacy budget that satisfies privacy requirements while maximizing utility. A natural and well-studied workaround is to use personalized privacy budgets, which may differ across agents. In this paper, we show that personalized budgets come with major limitations and that for mean estimation, the dominant factor is not full personalization, but rather choosing the right effective privacy budget. This can be achieved through a simple thresholding operator that we describe. Compared with this thresholding baseline, the gains obtained by fully personalized mechanisms are limited. In particular, we precisely quantify the constant-factor improvement in settings with mixed private and public datasets and in private datasets with two levels of privacy requirements. We also establish upper bounds and identify regimes of maximal gain for arbitrary privacy requirements.
\end{abstract}

\section{Introduction}

Differential Privacy (DP) has become a gold standard for privacy protection, particularly for releasing population-level statistics and deploying machine learning models \cite{Dwork2006}. Differential Privacy ensures that the presence or absence of any single data point in the training data cannot affect the final model too much, thereby guaranteeing that an outside observer cannot learn too much about a specific record from the resulting model. This guarantee is controlled by a parameter $\varepsilon$, called the privacy budget: a smaller $\varepsilon$ corresponds to a stronger privacy guarantee. Enforcing DP requires randomization, often achieved by injecting noise into the computation. This noise hides an agent's data: variations in the model or statistic when one's data is or is not included can now be attributed to the noise added to the computation, while complying with all reported privacy levels. However, it also degrades the utility and accuracy of the computation. A recurring question in DP is thus how to select the right privacy budget, in order to achieve the highest possible accuracy while maintaining satisfactory privacy constraints \cite{Ponomareva2023, Cyffers2025SettingI}.

In the classical formulation of Differential Privacy, a single privacy budget is imposed uniformly across all participants, i.e., all data points are guaranteed the same level of privacy $\varepsilon$. A natural extension is therefore to relax this assumption and allow individualized privacy budgets. From an optimization point of view, this enlarges the space of algorithms satisfying the privacy constraints and can therefore remove a source of suboptimality, in the same way that personalized pricing can improve profitability. This personalization of the privacy budget is also grounded in the well-studied fact that different agents have different privacy attitudes \cite{Allen2011-ALLUPW, Gerber2018, contextual}. Exploiting individualized privacy budgets has thus emerged as a rich research field. In particular, this has been observed by \cite{Chaudhuri2023MeanEU, Cummings2022OptimalDA, Fallah2022OptimalAD, Syomantak} for mean estimation. The affine estimator that assigns different levels of privacy to different agents in the central model through reweighting, also known as a sensitivity pre-processing step \cite{IndividualSensitivityPreprocessing}, vastly improves over the naive baseline that sets the noise level according to the most stringent privacy requirement. More precisely, these works i) show that when agents can have one of two possible privacy requirements, their affine estimator is optimal, ii) compute the optimal weights and affine estimator that minimize the mean-squared error of the final estimate for \emph{any} potential privacy profile, including, e.g., the case where every single agent has a different privacy requirement, and iii) characterize the mean-squared error of the resulting estimator. %

However, choosing the estimator that uses all data points with the minimum privacy budget as a baseline may be too naive. Consider a setting with $n$ agents. One agent $a$ has a very stringent privacy budget, e.g., $\eps_a = 0.01$, while all remaining agents have moderate privacy requirements, e.g., $\eps = 1$. In this case, a learner should not add noise calibrated to the smallest privacy budget $\eps_a$ to the entire computation. Instead, they should discard the data of agent $a$ and tailor the rest of the mechanism to the much more favorable privacy level $\eps$, thereby largely reducing the amount of noise added to the computation. This suggests that another natural baseline against which to compare fully personalized mechanisms is a single privacy budget $\eps$ for all retained data points, where $\eps$ is chosen more carefully than the minimum. This leads to a natural thresholding estimator: given a set of privacy preferences from agents $1$ to $n$, denoted by $\eps_1, \ldots, \eps_n$, the learner chooses a threshold $\tau$, discards the data points of all agents $i$ such that $\varepsilon_i < \tau$, and adds noise calibrated to privacy level $\tau$ to the remaining dataset---perhaps the simplest form of privacy techniques that trim part of the dataset, as seen in~\cite{bun2019average}. This is a transparent and easy-to-implement heuristic for dealing with agents with widely differing privacy preferences, and it is not hard to see that this estimator can be orders of magnitude better than the naive minimum-budget estimator.

However, this raises the following question: in the simple problem of mean estimation, how much is actually gained from personalization? The thresholding operator has several advantages over personalized mechanisms: it enables easier auditing of the privacy mechanism, clearer communication of the algorithm's privacy guarantee, and simpler adaptation to changes in the set of participating users. It also provides privacy for free to the agents who were more generous with their data. Yet, from a practical perspective, it is not clear whether the best-known affine estimator yields significant gains in the privacy-accuracy trade-off compared to the thresholding estimator proposed above and in the related literature. In particular, to the best of our knowledge, no prior work has attempted to bound the gap between this simple thresholding heuristic and the optimal affine estimator of the population mean. The question we ask is therefore the following: does personalized affine weighting provide a genuinely meaningful improvement over such thresholding, or are the gains modest in most, or even all, regimes of privacy preferences? Our results show that the gains are quite modest in most practical scenarios. 

\paragraph{Summary of contributions} Our contributions are as follows:
\begin{itemize}
    \item \textbf{Public and private data (Section \ref{sec:public_case}).} We first study the setting where private data can be combined with public data, and show that our simple threshold-based estimator achieves a $2$-approximation to the optimal affine estimator.
    
    \item \textbf{Two privacy levels (Section \ref{sec:two-level}).} We then turn to the case of two finite privacy levels. We first show a gap with the case of a public and a private dataset, and show that one cannot achieve a $2$-factor approximation in the worst case by using our threshold-based estimator. However, we show that the threshold-based estimator continues to provide a constant-factor guarantee, namely a slightly worse $4$-approximation to the optimal affine estimator.
    
    \item \textbf{Arbitrary privacy levels (Section \ref{sec:general_case}).} Finally, we consider the general case with an arbitrary number of privacy levels, where we show that no constant-factor approximation is possible in general: if $m$ privacy levels are allowed among $n$ data points, we show that the approximation ratio between the threshold-based and the optimal affine estimator is tightly given by $\Omega \left( \min \left(\log^2 n, m^2\right)\right)$.
\end{itemize}

\paragraph{Related Work}

While DP typically imposes uniform privacy budget across all individuals, a substantial line of research considers personalized privacy guarantees. %
In particular, the notion of \emph{heterogeneous differential privacy} that we use in this work was originally formalized by \cite{GhoshRoth2011,Jorgensen2015,Alaggan2015}. This definition refines the traditional notion of DP to allow each individual or data point to be protected with a distinct privacy parameter. A core idea that allows personalized privacy guarantees in the central privacy model is ``sensitivity preprocessing''~\cite{IndividualSensitivityPreprocessing}---simply said, allowing queries to have different sensitivities with respect to different data points. The lower the sensitivity of a given individual or data point in a computation, the better their privacy guarantee. For mean estimation, \cite{Chaudhuri2023MeanEU, Syomantak} leverage this idea---in particular, giving different weights to different data points as a function of their privacy requirement---and study how allowing heterogeneous privacy levels can improve accuracy compared to enforcing the most stringent privacy requirement uniformly across all users. Beyond mean estimation, there has been a small line of work aiming to understand heterogenous differential privacy in loss minimization contexts~\cite{boenisch2022individualized,boenisch2023have,persoRidgejuba}. There has also been a significant line of work in mechanism design for data acquisition and data markets~\cite{juba_itcs15,arpita_katrina,GhoshRoth2011,Cummings2022OptimalDA, Fallah2022OptimalAD,juba_2026}; in these settings, agents incur heterogeneous economic costs for privacy, and must be compensated for said costs either through a monetary payment or a service offered in return. There has also been work on personalized privacy accounting, where different data points may incur different realized privacy losses over time, and the goal is to track the effective privacy budget consumed by each data point~\cite{Feldman2020}. Finally, a complementary and important line of work, whose goal is also to improve privacy-accuracy trade-offs in DP, studies how public or non-sensitive data can be leveraged. These works show that auxiliary public data can significantly reduce the cost of privacy in estimation and optimization tasks \cite{publicprivate19, Wang2020, NEURIPS2022_765ec499, NEURIPS2024_243697ac, 10.5555/3692070.3693404, Block2024OracleEfficientDP}. 

Most of these works assume that privacy preferences and data are uncorrelated, in particular due to strong impossibility results~\cite{GhoshRoth2011,nissim2014redrawing}. However, very recently,~\cite{Chaudhuri2025bis} has made the first breakthrough in over a decade in allowing heterogeneous DP while accounting for the correlation between data points and privacy requirements; this setting is outside the scope of the current paper but an important foundation for future work.

Despite this rich literature, existing work primarily focuses on designing or characterizing optimal mechanisms under heterogeneous privacy constraints. In contrast, relatively little is understood about the marginal value of full personalization compared to simpler alternatives. In particular, prior analyses typically benchmark personalized mechanisms against naive baselines that calibrate noise to the most stringent privacy requirement, which can significantly underestimate the performance of simpler approaches. In this work, we revisit this question in the context of mean estimation by comparing the best-known affine estimators under heterogeneous privacy with a simple thresholding approach that enforces a single privacy level on a subset of the data. Our results show that, in most regimes of interest, the gains from full personalization are relatively limited.

\section{Preliminaries for Differentially Private Mean Estimation}

In this section, we introduce and motivate the definitions and tools used in the remainder of the paper.

\begin{definition}[Approximate Differential Privacy]
  A mechanism $\M$ is $(\eps, \delta)$-differentially private if, for all neighboring datasets $D \sim D'$ and all measurable sets $S$, the following inequality holds, where the randomness is over the randomness of $\M$:
\[
\Pr(\M(D)\in S) \leq e^{\eps} \Pr(\M(D')\in S) +\delta.
\]
\end{definition}

In this definition, the dataset $D$ is typically a set of records $D=\{x_i\}_{i=1}^n$, and the neighboring relation corresponds to replacing one record by another. For instance, $D' = \{x'_i\}_{i=1}^n$ with $x_1 \neq x'_1$ and $x_i = x'_i$ for all $i>1$. Pure differential privacy corresponds to the special case where $\delta = 0$. A simple DP mechanism is the Laplace mechanism.

\begin{definition}[Laplace mechanism]
  Let $\Delta = \max_{x \sim x'}|f(x) - f(x')|$. The Laplace mechanism for the function $f$ is
  \[
  \M(x) = f(x) + \eta, \quad \eta \sim \Lap(\Delta/\eps),
  \]
  and this mechanism is $(\eps, 0)$-DP.
\end{definition}

For personalized privacy budgets, a first attempt is to include the privacy budget within each record and consider $D= \{(x_i, \eps_i, \delta_i)\}_{i=1}^n$, allowing arbitrary changes in the triplet as before.

\begin{definition}[Arbitrary Personalized DP]
  A mechanism $\M$ is arbitrary personalized differentially private if, for all neighboring datasets $D \sim D'$, where $(x, \eps, \delta)$ is changed into $(x', \eps', \delta')$, and all measurable sets $S$, the following inequality holds, where the randomness is over the randomness of $\M$:
\[
\Pr(\M(D)\in S) \leq e^{\min(\eps, \eps')} \Pr(\M(D')\in S) +\min(\delta, \delta').
\]
\end{definition}

Unfortunately, this definition is too stringent. Let $\eps^* = \min_{1\leq i\leq n} \eps_i$ and $\delta^* = \min_{1\leq i\leq n} \delta_i$. A well-known folklore result attributed to Kobbi Nissim, Salil Vadhan, and David Xiao\footnote{While the result is well known to the differential privacy community, there is no public source formally stating the result that can be cited here. We provide a proof for completeness.} states that any arbitrary Personalized DP algorithm is $(2\eps^*, (1+e^{\eps^*})\delta^*)$-DP. In other words, the fact that some records are satisfied by larger privacy budgets does not allow the use of a less restrictive mechanism. %

\begin{proof}
  Fix $D \sim D'$ differing at the $i$-th record, and denote by $(x_i, \eps_i, \delta_i)$ and $(x'_i, \eps'_i, \delta'_i)$ their respective $i$-th records. Both datasets are also neighboring from $D''$, where all records are equal to those of $D$ except the $i$-th one, replaced by $(x, \eps^*, \delta^*)$. It follows that the two following inequalities must hold:
\[
\Pr(\M(D)\in S) \leq e^{\eps^*} \Pr(\M(D'')\in S) +\delta^*
\]
and
\[
\Pr(\M(D'')\in S) \leq e^{\eps^*} \Pr(\M(D')\in S) +\delta^*.
\]
Combining both gives the result.
\end{proof}

\begin{remark}
  This impossibility result is not specific to mean estimation, but holds for any differentially private algorithm.
\end{remark}

Thus, it is not possible to exploit personalized privacy budgets if these budgets can be changed arbitrarily: in that case, one is forced to align the mechanism with the most stringent privacy requirement. Taking advantage of personalization already assumes that the learner can access the privacy budget of each participant. This can be achieved by decoupling the privacy budget from the data, as is done in most related work~\cite{GhoshRoth2011,Jorgensen2015,Fallah2022OptimalAD,IndividualSensitivityPreprocessing,Cummings2022OptimalDA,Syomantak,Chaudhuri2023MeanEU}.

\begin{definition}[Heterogeneous Differential Privacy~\cite{GhoshRoth2011,Jorgensen2015}]\label{def:heterogenenous_DP}
  Fix a vector $\veps = (\eps^{(1)}, \dots, \eps^{(n)})$. The mechanism $\M$ is heterogeneous $\veps$-DP if, for all $i \in[n]$,
\[
\Pr(\M(D) \in S) \leq e^{\eps^{(i)}} \Pr(\M(D^{\prime i}) \in S),
\]
for all measurable sets $S$, where $D$ and $D^{\prime i}$ are any two neighboring datasets that differ arbitrarily only in the $i$-th component.
\end{definition}
We adopt this definition in the rest of the paper. 

\section{Model}
In this work, we compare the performance of two estimators, defined formally below: one with heterogeneous privacy budgets and one with homogeneous privacy guarantees.
We consider the task of estimating the mean $\mup$ of a distribution $\mathcal{P}$ with bounded support $[-1/2,1/2]$. The learner has access to $n$ data points drawn i.i.d. from $\mathcal{P}$. In addition, each data point $x_i$ comes with its own privacy requirement. We assume that the privacy budgets $\{\eps^{(i)}\}_i$ take values in a finite set of at most $m \leq n$ distinct elements, denoted by $\eps_1 < \ldots < \eps_m$, and we denote by $n_i$ the number of participants with privacy budget $\eps_i$. To ensure heterogeneous DP, any estimator that uses data point $x_i$ must satisfy $\eps^{(i)}$-differential privacy with respect to $x_i$.

\subsection{Mean estimation problem}

Given an estimator $\hat\mu$ of the population mean $\mup$, we define the distributionally-robust \emph{risk}, or \emph{mean-squared error}, as our main measure of the accuracy, or quality, of an estimator:
\begin{definition}[Distributionally-Robust Mean-Squared Error]
\[
\sup_\mathcal{P} \mathbb{E}\big[(\hat\mu - \mu_\mathcal{P})^2\big].
\]
\end{definition}
Our work quantifies the maximum possible gains on this metric obtained by fully tailoring the privacy mechanism to personalized privacy budgets. We do so by comparing the best affine fully personalized estimator with an estimator that ensures a single privacy budget.

\subsection{Best affine estimator}

 We consider the class of unbiased affine estimators, i.e. estimators based on a noisy weighted average of the data points
\[
\hat\mu_w(D)=\sum_{i=1}^n w_i x_i + Z\quad ,
\]
where $w_i\ge 0$, $\sum_{i=1}^n w_i = 1$, and $Z\sim \Lap(\eta)$. To ensure $\veps$-heterogeneous differential privacy for every user $u$, the noise scale must satisfy
\[
\eta \ge \max_{i\in[n]} \frac{w_i}{\eps^{(i)}}.
\]
Thus, for fixed weights, the optimal noise parameter is
\[
\eta = \max_{i\in[n]} \frac{w_i}{\eps^{(i)}},
\]
and designing an optimal estimator boils down to finding the optimal weights. Such a characterization of the optimal weights is given in \cite{Syomantak}, and we refer the reader to that work for detailed expressions. For our purposes, we only need to characterize the mean-squared error of the optimal affine estimator. The error comes from two sources: the variance of the data points themselves and the variance of the injected noise. The optimal affine estimator corresponds to the weights that minimize these two terms, leading to the following formula:
\begin{equation}
\mse_{\mathrm{aff}}(\eps)
:=
\inf_{\substack{w_i\ge 0\\ \sum_{i=1}^n w_i=1}}
\left\{
\frac14 \sum_{i=1}^n w_i^2
+
2\left(\max_{i\in[n]} \frac{w_i}{\eps^{(i)}}\right)^2
\right\}.
\label{eq:mseafftedious}
\end{equation}

Several previous works \cite{Syomantak,Cummings2022OptimalDA,Fallah2022OptimalAD} obtain a more explicit characterization of this error by noting that the optimal weights are obtained by clipping the privacy levels. More precisely, the optimal affine estimator sets a threshold $\tau > 0$, defines $\bar\eps^{(i)} = \min\{\eps^{(i)},\tau\}$, and ensures privacy with respect to these truncated privacy levels. Since $\bar\eps_i \leq \eps_i$, the privacy requirements of all data points are satisfied. The weights are then set to be proportional to the clipped privacy budgets and normalized, namely 
\[
w_i = \frac{\min\{\eps^{(i)},\tau\}}{\sum_i \min\{\eps^{(i)},\tau\}}\]
and the noise parameter is set as before to satisfy all privacy constraints: $\eta = \frac{1}{\sum_i \min\{\eps^{(i)},\tau\}}.$
Using this set of weights yields an explicit formula for the mean-squared error of the optimal affine estimator.

\begin{claim}\label{clm:threshold_characterization}
For any $\tau>0$, define 
\[
s_{\tau}:=\sum_{i=1}^n \min\{\eps^{(i)},\tau\},
\qquad
q_{\tau}:=\sum_{i=1}^n \min\{\eps^{(i)},\tau\}^2.
\]
Then
\begin{equation}
\mse_{\mathrm{aff}}(\eps)
=
\inf_{\tau>0}
\left(
\frac{q_{\tau}}{4s_{\tau}^2}
+
\frac{2}{s_{\tau}^2}
\right).
\label{eq:mseaff}
\end{equation}
\end{claim}

The proof is straightforward and can be found in Appendix~\ref{app:claim-1}. When the data points can have at most two privacy requirements, i.e., $m = 2$ with $0 < \eps_1 < \eps_2$, \cite{Syomantak} shows that the affine estimator above is optimal in terms of mean-squared error. An immediate extension is that the affine estimator is also optimal when combining private data at a predetermined privacy level $\eps_1 > 0$ with public data, corresponding to taking $\eps_2 \to \infty$. When $m > 2$ heterogeneous privacy levels are allowed, the optimal \emph{affine} estimator may no longer be optimal among all estimators. However, we still use the optimal affine estimator as our main point of comparison, noting that it is currently the best-known estimator for mean estimation under heterogeneous privacy.

\subsection{Unique-Threshold $\eps$-estimator}

We compare the best affine estimator with a simple and natural estimator, which we call the \emph{unique-threshold $\eps$ operator}. The key idea is to select a single privacy budget that respects the guarantees imposed by $\veps$-heterogeneous differential privacy while minimizing the mean-squared error. Informally, this operator discards all data points with overly stringent privacy requirements, then adds noise calibrated to the remaining dataset. This estimator does not exploit heterogeneity: it provides a single, identical privacy level $\eps$ to all agents whose data is used in the mean estimation. More formally, we define the estimator as follows.

\begin{definition}[Unique-threshold $\eps$ estimator]
  Let $\veps$ be fixed. For a given $\eps$, let $\uno(\eps)$ be the vector whose $i$-th component is $1$ if $\eps^{(i)} \geq \eps$ and $0$ otherwise, and let $n_{\eps} = \sum_{i=1}^{n} \uno(\eps)_i$. The \emph{threshold mechanism} $\M_{\eps}$ is defined by
  \[
  \M_{\eps}(D) = \frac{1}{n_{\eps}} \sum_{i=1}^{n} \uno(\eps)_i x_i + \Lap \left(\frac{\Delta}{n_{\eps} \eps}\right),
  \]
  and is $\veps$-DP. The \emph{unique-threshold $\eps$ estimator} corresponds to a choice of $\eps$ in the threshold mechanism minimizing the mean squared error.
\end{definition}

As for the affine estimator, the mean-squared error is given by
\begin{equation}
\mse_{\thr}(\veps)
:=
\inf_{\eps>0}
\left\{
\frac{1}{4n_\eps}+\frac{2}{\eps^2 n_\eps^2}
\right\}.
\label{eq:msethr}
\end{equation}
The rest of the paper compares these two estimators under different scenarios for the privacy-budget vector $\veps$.

\section{Limits of public-private combination}\label{sec:public_case}

\begin{figure}[bh]
    \centering

    \begin{subfigure}[t]{0.35\linewidth}
        \centering
        \includegraphics[width=\linewidth]{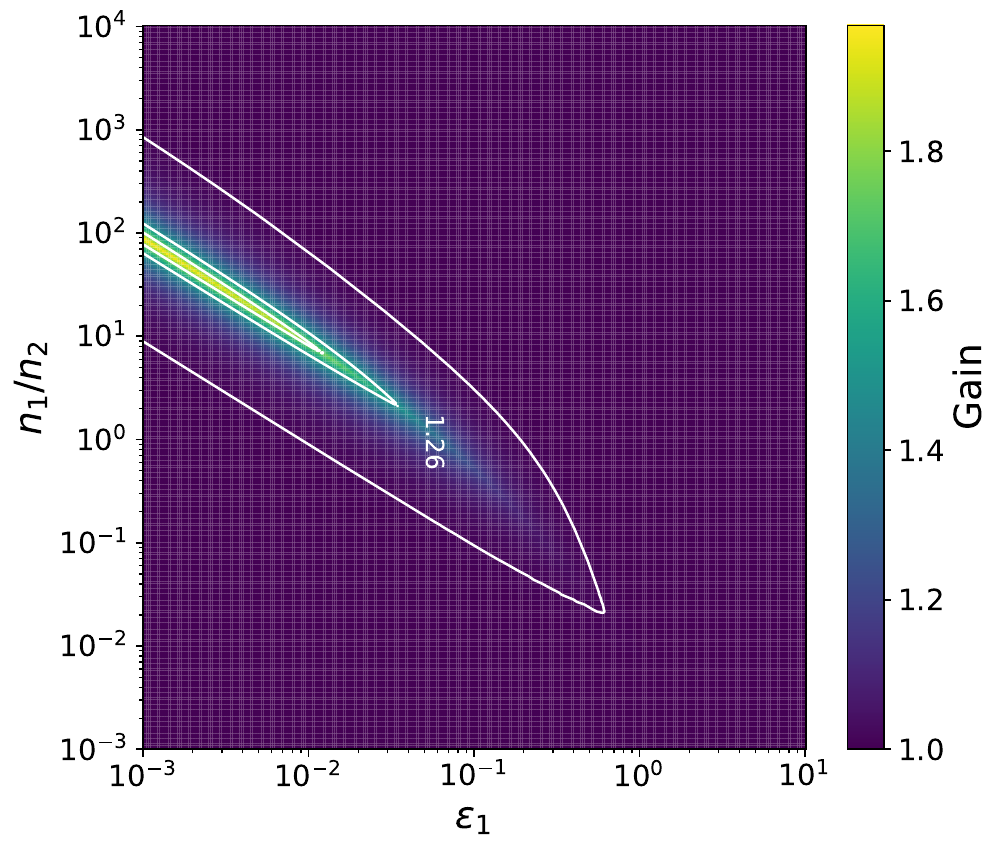}
        \caption{}
        \label{fig:ppheatmap}
    \end{subfigure}
    \hfill
    \begin{subfigure}[t]{0.3\linewidth}
        \centering
        \includegraphics[width=\linewidth]{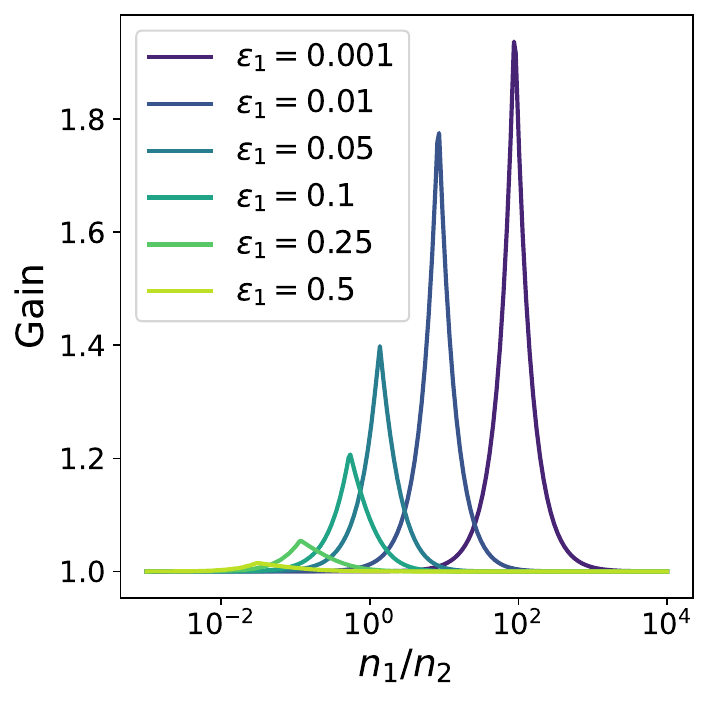}
        \caption{}
        \label{fig:ppn}
    \end{subfigure}
    \hfill
    \begin{subfigure}[t]{0.3\linewidth}
        \centering
        \includegraphics[width=\linewidth]{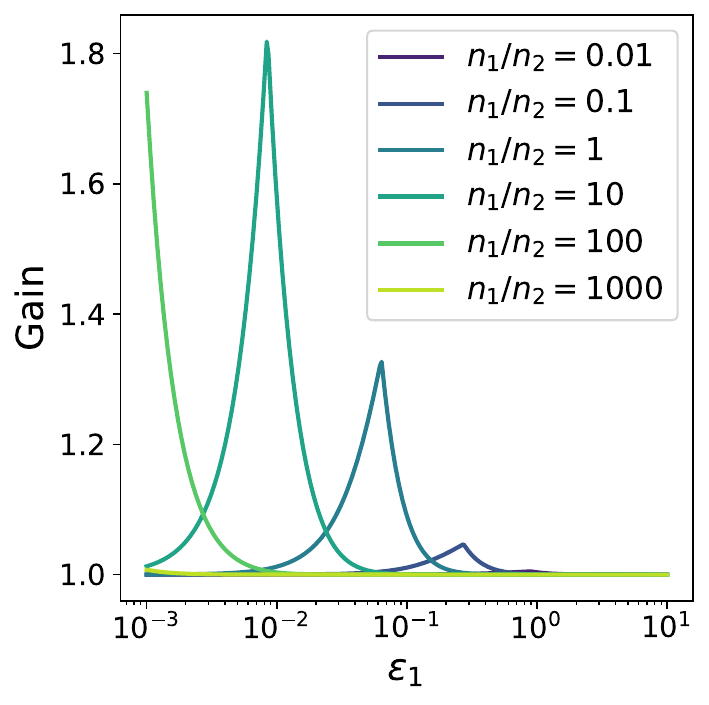}
        \caption{}
        \label{fig:ppe}
    \end{subfigure}

    \caption{Ratio between the unique-threshold estimator and the best affine operator for a combination of public and private data, for a fixed $n_2 = 1000$. The heatmap shows a narrow region where the gain is significant. The outermost contour level is set to $1.01$. Figures \ref{fig:ppe} and \ref{fig:ppn} correspond respectively to rows and columns of the heatmaps, making it possible to visualize the rapid decay of the gain.}
    \label{fig:three-side-by-side}
\end{figure}

We first consider a scenario in which a mixture of public and private data is available. The learner has access to $n_1$ data points with privacy requirement $\varepsilon_1 > 0$. In addition, the learner has access to $n_2$ public data points, for which no noise needs to be added, equivalently, $\varepsilon_2 = +\infty$. In this regime, the public data do not require noise injection to estimate the mean. Thus, for the private and public parts of the estimator to have comparable variance, one typically needs $n_1 \gg n_2$. At this equilibrium, the two parts contribute equally to the error, leading to a constant factor of $2$.

\begin{theorem}\label{thm:public-data-factor-2}
Let $\mse_{\thr}$ denote the risk of the unique-threshold estimator in the public-data regime, and let $\mse_{\mathrm{aff}}$ denote the optimal affine risk. Then
\[
\frac{\mse_{\thr}}{\mse_{\mathrm{aff}}} \le 2.
\]
\end{theorem}

The proof can be found in Appendix~\ref{app:public}. This bound is tight and can be approached, when
\[
n_2 \approx \frac{n_1}{1 + 8/(n_1\eps_1^2)}.
\]
For example, the ratio is close to $1.95$  for $\eps_1 = 0.001$, $n_1 = 10^4$, and $n_2 = 12$. The ratio however drops quickly to $1$ when the proportions are not following this rule, as reported in Figure~\ref{fig:three-side-by-side}.

\section{A warm-up: two privacy budgets}\label{sec:two-level}
In this section, we adopt nearly the same setting as in the previous section, but with finite privacy for $\eps_2$. Thus, the learner has access to $n_1$ data points for which they must satisfy $\eps_1$-differential privacy, and to $n_2$ data points with an $\eps_2$-differential privacy requirement, where $0 < \eps_1 < \eps_2 < +\infty$. In this case, the approximation guarantee degrades compared to Section~\ref{sec:public_case}, from a factor of $2$ to a factor of $4$. We first give a simple counterexample violating the factor-$2$ guarantee, and then prove the factor-$4$ bound.

\begin{example}
Consider $\veps = (\frac{1}{2}, 1)$. We can compute the mean-squared errors of the two estimators using Equations~\eqref{eq:mseaff} and \eqref{eq:msethr}. It is easy to verify that the unique-threshold estimator has mean-squared error $17/8$, while the optimal affine estimator has mean-squared error $37/36$. In particular, this implies that
\[
\frac{\mse_{\mathrm{th}}(\veps)}{\mse_{\mathrm{aff}}(\veps)}
\ge
\frac{17/8}{37/36}
>2.
\]
\end{example}

However, we show that with two privacy levels, the approximation ratio between the unique-threshold estimator and the optimal affine estimator remains constant, independently of the value of $\veps$.

\begin{theorem}
\label{thm:two-level-factor-4}
Let $\mse_{\mathrm{th}}(\veps)$ denote the mean-squared error of the unique-threshold estimator, and let $\mse_{\mathrm{aff}}(\veps)$ denote the optimal affine risk. Then
\[
\frac{\mse_{\mathrm{th}}(\veps)}{\mse_{\mathrm{aff}}(\veps)} \le 4.
\]
\end{theorem}

The full proof can be found in Appendix~\ref{app:two-level}. It follows a similar structure to the public-private case, with a slightly more involved characterization of the transition between the regime where $\tau = \eps_1$ is optimal and the regime where $\tau = \eps_2$ is optimal. Note that this factor corresponds to the improvement obtained from having four times more data, or from doubling $\eps$ in regimes where $\eps$ is small enough relative to $n$.

\begin{figure}[bh]
    \centering

    \begin{subfigure}[t]{0.35\linewidth}
        \centering
        \includegraphics[width=\linewidth]{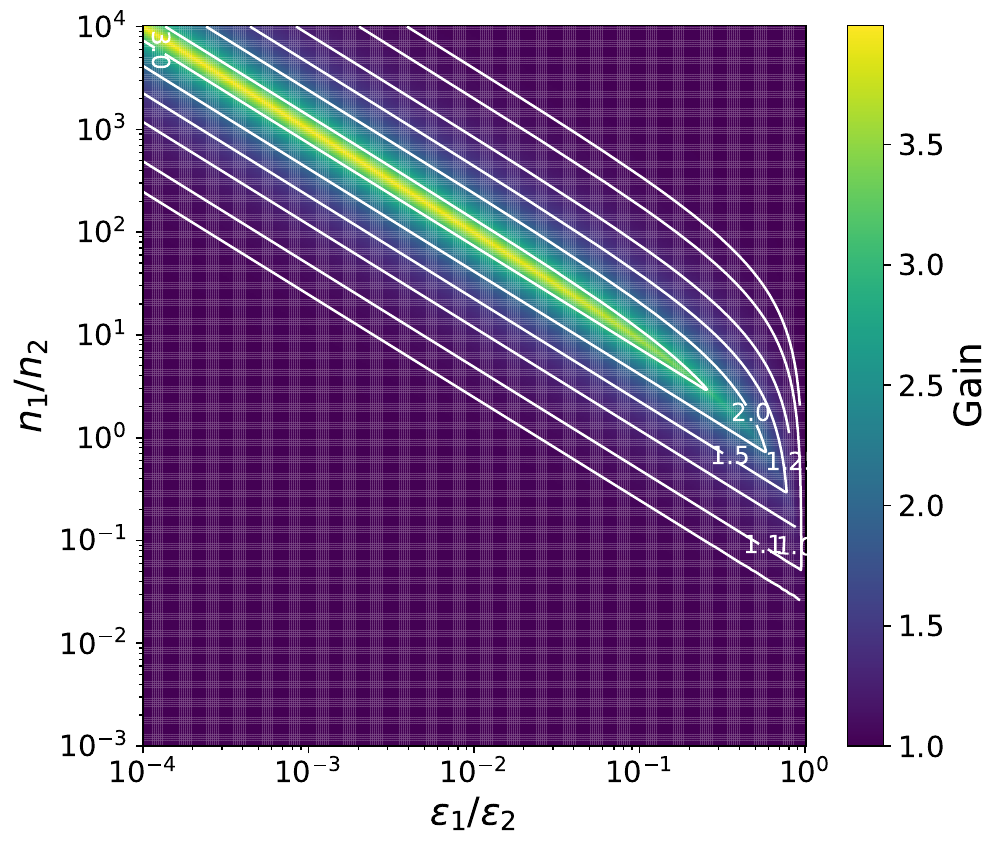}
        \caption{}
        \label{fig:heatmap}
    \end{subfigure}
    \hfill
    \begin{subfigure}[t]{0.3\linewidth}
        \centering
        \includegraphics[width=\linewidth]{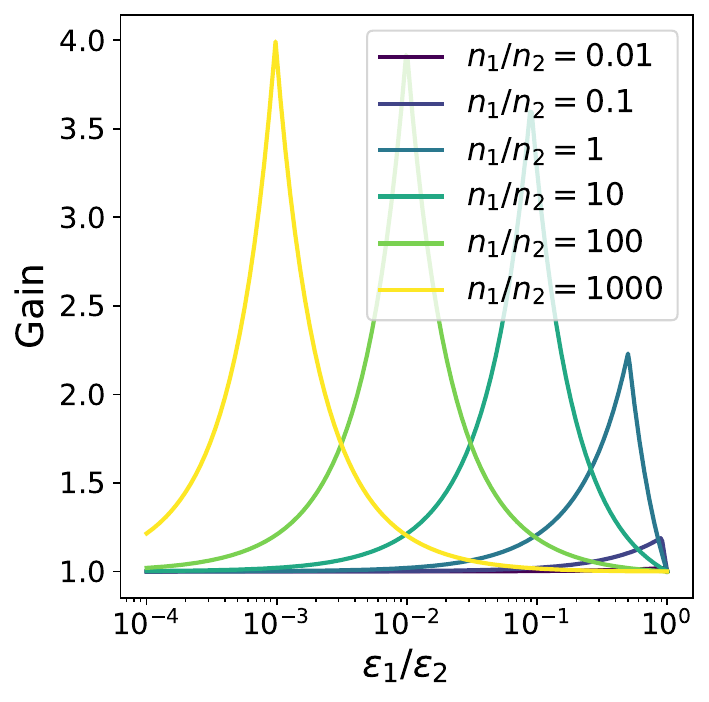}
        \caption{}
        \label{fig:n}
    \end{subfigure}
    \hfill
    \begin{subfigure}[t]{0.3\linewidth}
        \centering
        \includegraphics[width=\linewidth]{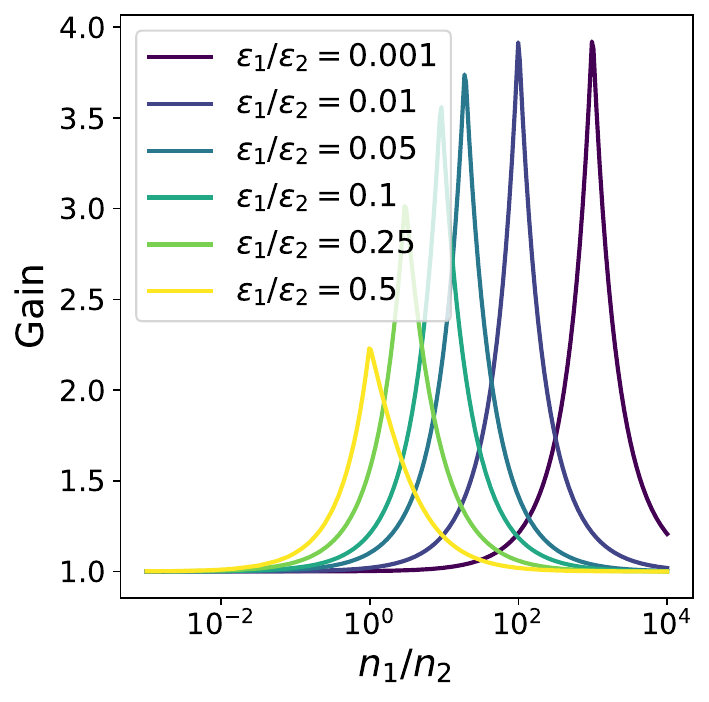}
        \caption{}
        \label{fig:e}
    \end{subfigure}

    \caption{Ratio between the unique-threshold estimator and the best affine operator for two finite privacy levels, similarly to Figure \ref{fig:three-side-by-side}, keeping $n_2 = 1000$ and $\eps_2 = 0.03$ to ensure a large gain ratio.}
    \label{fig:all-private}
\end{figure}

\section{Limits of General choice of privacy budgets}\label{sec:general_case}

Finally, we consider the general case in which there is an arbitrary number of privacy levels $m \geq 2$. For all $i \in [m]$, we denote by $n_i$ the number of data points with privacy requirement $\eps_i$, where $0 < \eps_1 < \ldots < \eps_m$. We denote by $n$ the total number of available data points across all privacy levels.

We show that \textbf{the unique-threshold operator cannot guarantee a constant-factor approximation to the best-known affine estimator.} We first provide a lower bound showing that no constant-factor approximation can be guaranteed when the number of privacy levels is not constant. In particular, we construct a class of privacy parameters for which the ratio grows with $m$, showing that the approximation ratio is not uniformly bounded by a constant independent of the problem parameters.

\begin{theorem}
\label{thm:lower-bound-general-m-levels}
Take any $m \geq 2$. There exist privacy levels $0 < \eps_1 < \cdots < \eps_m$ and data quantities $n_1,\dots,n_m \in \mathbb{N}$ with $n \triangleq \sum_{i=1}^m n_i = 2^m-1$ such that
\[
\frac{\mse_{\thr}(\veps)}{\mse_{\mathrm{aff}}(\veps)}
\;\ge\;
\frac{m^2}{5}.
\]
In particular, as $m = \log_2 (n+1)$, we obtain
\[
\frac{\mse_{\thr}(\veps)}{\mse_{\mathrm{aff}}(\veps)}
\geq \frac{(\log_2 (n+1))^2}{5}.
\]
\end{theorem}

\begin{proof}[Proof Sketch]
The full proof is given in Appendix~\ref{app:general-case}. The intuition behind our construction is as follows. The unique-threshold estimator chooses a privacy threshold $\eps$: a smaller $\eps$ requires the learner to satisfy more stringent privacy constraints, but also allows them to use a larger number of data points $n_\eps$. Since the sensitivity of the mean is $1/n_\eps$, as recalled in our preliminaries, the amount of noise required for privacy is controlled by
$\Delta/\eps = 1/n_\eps\eps$.
We construct an instance of the problem in which: i) $n_\eps \eps = O(1)$ regardless of how the learner chooses $\eps$, leading to a constant noise level and hence to a constant mean-squared error for the unique-threshold estimator; while ii) the optimal affine estimator can still exploit all privacy levels, and its variance decreases as $m$ increases. Concretely, our construction \emph{saturates} each privacy level by taking $\eps_i = 2^{-(i-1)}$ and $n_i = 2^{i-1}$, ensuring that $n_i \eps_i$ and $\eps n_\eps$ remain of order $O(1)$ regardless of the threshold chosen by the learner\footnote{This construction is similar to a well-known distribution in mechanism design called the equal revenue distribution~\cite{hartline2009simple}, which ensures that, given a cdf $F$, $x (1-F(x)) = 1$ for all $x$. In mechanism design, this distribution induces situations where a posted price mechanism always extracts the same revenue of $1$ independently of the choice of price, and is particularly relevant in bounding the gap between simple and optimal mechanisms in single-bidder single-item auctions.}; this forces the learner to add total noise scaling with $O(1)$. At the same time, the affine optimum assigns smaller weights to more private users and larger weights to less private ones, splitting the total privacy burden of $\sum_{i=1}^m n_i \eps_i \sim m$ unevenly across agents. In particular, setting the weight 
\[w_i = \frac{\eps_i}{\sum_{i=1}^m n_i \eps_i} \sim \frac{\eps_i}{m} \quad \text{and} \quad \eta = \frac{1}{\sum_{i=1}^m n_i \eps_i} \sim \frac{1}{m}\]
guarantees the desired level of privacy and ensures the privacy noise parameter scales with $1/m$, leading to a variance due to noise scaling with $1/m^2$. 
\end{proof}

\paragraph{An upper bound on the approximation factor}
Finally, we conclude our results by showing that this bound is effectively tight. Namely, we show that for any instance with $m$ distinct privacy levels, the approximation factor of the unique-threshold estimator is at most $\min\!\bigl\{(1+\log_2 n)^2,\; m^2\bigr\}$.

\begin{theorem}
\label{thm:upper-bound-general-m-levels}
For any $m \geq 2$, for any $0 < \eps_1 < \dots < \eps_m$ and $n_1,\dots,n_m \in \mathbb{N}$ such that $n = \sum_{i=1}^m n_i$, we have that: 
\[
\frac{\mse_{\thr}(\veps)}{\mse_{\mathrm{aff}}(\veps)}
\;\le\;
\min\!\bigl\{(1+\log_2 n)^2,\; m^2\bigr\}.
\]
\end{theorem}

\begin{proof}[Proof Sketch]
The proof of the Theorem can be found in Appendix~\ref{app:general-case}. The proof relies on the explicit MSE characterization for the affine estimator given in \ref{eq:mseaff}. Fix any $\tau>0$, and let $\bar\eps_j=\min\{\eps_j,\tau\}$. Then the affine estimator corresponding to $\tau$ has noise controlled by the total clipped quantity
\[
s_\tau=\sum_{j=1}^m n_j \bar\eps_j.
\]
The unique-threshold estimator is controlled by a directly comparable quantity, namely $\eps^* n_{\eps^*}$. For every fixed clipping constant $\tau$, one can always choose a threshold $\eps^*$ among the clipped values such that
\[
\eps^* n_{\eps^*} \geq \frac{s_\tau}{m},
\]
by the pigeonhole principle. In other words, although the affine estimator spreads its contribution across several clipped privacy levels, at least one threshold always captures at least a $1/m$ fraction of the total $s_\tau$. Since the threshold estimator adds noise of scale $1/(\eps^* n_{\eps^*})$, while the affine estimator corresponding to $\tau$ adds noise of scale $1/s_\tau$, this implies that the threshold estimator loses at most a factor $m$ in noise scale, and therefore at most a factor $m^2$ in the resulting variance. The worst-case of $(1 + \log_2 n)^2$ comes from the harmonic series $H_n \sim 1 + \log_2 n$, and a standard argument that $s_\tau \leq H_n \max_\eps \eps n_\eps$. 
\end{proof}

Note that the bounds of Theorems~\ref{thm:lower-bound-general-m-levels} and~\ref{thm:upper-bound-general-m-levels} are tight up to constant factors. More precisely, whenever $2^m-1 \le n$, both bounds are of the order of $\Omega(m^2)$.  Once $m$ becomes larger than $\log_2 n$, the approximation factor saturates at order $\log_2^2 n$ and increasing $m$ cannot worsen the ratio as per Theorem~\ref{thm:upper-bound-general-m-levels}. i.e., the approximation factor grows quadratically with the number of privacy levels so long as $2^m \leq n$.

\section{Discussion}\label{sec:discussion}

Our work focuses on the specific problem of mean estimation, which limits the generality of our results. Indeed, in the context of mean estimation, affine estimators are optimal for $m = 2$; however, it is unclear whether they remain optimal for $m > 2$. A natural direction is therefore to characterize the optimal estimator in this regime and quantify its gap with the simple threshold-based approach. Second, while mean estimation is a canonical task and a standard proxy in the literature, extending these insights to broader learning and optimization problems remains an important challenge. A key limitation is that there is currently no theoretical characterization of near-optimal or optimal heterogeneous differentially private estimators, even for simple problems such as linear regression. Mean estimation is often used as a proxy for harder tasks. For instance, the widely used DP-Follow-the-Regularized-Leader SGD minimizes the mean-squared error over the sum of gradients, even though minimizing this quantity does not translate directly into optimization guarantees \cite{Pillutla2025CorrelatedNM}. We thus believe that, given the current state of the theory, mean estimation is the best available proxy for testing the advantage provided by full personalization.

Another open direction is whether our takeaways are specific to the Laplace mechanism and pure differential privacy. Our results primarily rely on how privacy constraints translate into the variance of an additive noise term, suggesting that similar qualitative behavior, where the gap is controlled by the number of privacy levels $m$, may extend to other additive mechanisms and alternative privacy definitions, such as z-CDP or approximate DP. However, extending our analysis to these settings is not immediate. In particular, privacy notions with more favorable composition properties may allow heterogeneous privacy constraints to aggregate differently, potentially reducing the overall noise required, for example by yielding sublinear dependence on $m$.

\section{Conclusion}\label{sec:conclusion}

Our results show that, in mean estimation under heterogeneous differential privacy constraints, relatively little is gained from full personalization. A simple thresholding estimator --- discarding overly private data and applying a single carefully chosen privacy level --- already captures most of the achievable utility compared to the optimal affine estimator. In particular, it achieves constant-factor approximations in natural regimes (e.g., a factor of 2 with public data and 4 with two privacy levels), and we precisely characterize when and how larger gaps can arise. When $n$ agents and $m > 2$ privacy levels are present, the picture becomes more nuanced, with an approximation factor of $O(\log^2 n)$ in the worst case; however, the gains remain modest.  Since optimal personalized estimators are not known in more complex settings, the gains from personalization for broader learning tasks remain open. We conjecture that similar takeaways could be established beyond mean estimation.

\section{Acknowledgments}

Prof. Ziani was supported by NSF CAREER IIS-2336236 and NSF Medium IIS-2504990. Edwige Cyffers was supported by the National Research Agency under France 2030, reference “ANR-23-IACL-0008”. This work was done while the authors were visiting the Simons Institute for the Theory of Computing.

\bibliographystyle{plain} 
\bibliography{biblio.bib}

\appendix

\input{appendix.tex}

\newpage

\end{document}

%% file: appendix.tex
\section{Preliminaries: Proof of Claim~\ref{clm:threshold_characterization}}\label{app:claim-1}

Since the weights sum to $1$, we must have
\[
1=\sum_{j=1}^n w_j=\sum_{j=1}^n \eta \min\{\eps^{(j)},\tau\},
\]
hence
\[
\eta=1/s_{\tau}.
\]
Therefore
\[
w_j=\frac{\min\{\eps^{(j)},\tau\}}{s_{\tau}},
\]
and the resulting worst-case mean-squared error is
\[
\frac14\sum_{j=1}^n w_j^2+ 2 \eta^2
=
\frac{q_{\tau}}{4s_{\tau}^2}+\frac{2}{s_{\tau}^2}.
\]
Taking the infimum over $\tau>0$ yields the claim.

In all the proofs, we use the notation $\mse_i$ to refer to the mean squared error obtained when using the threshold estimator of parameter $\eps_i$:
\[
\mse_{i}:=
\frac{1}{4n_{\eps_i}}+\frac{2}{(\eps_i n_{ \eps_i})^2}.
\]

\section{Proof of Theorem \ref{thm:public-data-factor-2}}\label{app:public}

We now provide the proof of the upper bound given in Theorem~\ref{thm:public-data-factor-2}. Let
\[
\mse_{\mathrm{priv}}=\frac{1}{4n_1}+\frac{2}{(n_1\varepsilon_1)^2}
\]
denote the risk of the optimal estimator that uses only the $n_1$ private samples, and let
\[
\mse_{\mathrm{pub}}=\frac{1}{4n_2}
\]
denote the risk of the empirical mean on the public samples. By definition,
\[
\mse_2=\mse_{\mathrm{pub}}.
\]

First, observe that using all $n_1+n_2$ samples at privacy level $\varepsilon_1$ can only improve over using the $n_1$ private samples alone, since both the sampling variance and the privacy noise decrease. Hence
\[
\mse_1 \le \mse_{\mathrm{priv}}.
\]
It follows that
\[
\mse_{\thr}
=
\min\{\mse_1,\mse_2\}
\le
\min\{\mse_{\mathrm{priv}},\mse_{\mathrm{pub}}\}.
\]

Next, in the public-data regime, the optimal affine estimator is the optimal linear combination of the private estimator and the public estimator as noted in~\cite{Syomantak}, and its risk is
\[
\mse_{\mathrm{aff}}
=
\frac{\mse_{\mathrm{priv}}\cdot \mse_{\mathrm{pub}}}
{\mse_{\mathrm{priv}}+\mse_{\mathrm{pub}}}.
\]
Indeed, if we combine two unbiased estimators with risks $a$ and $b$ using weights $\frac{b}{a+b}$ and $\frac{a}{a+b}$, the resulting risk is
\[
\frac{b^2 a + a^2 b}{(a+b)^2}
=
\frac{ab}{a+b}.
\]

Finally, for any $a,b>0$,
\[
\min\{a,b\} \le \frac{2ab}{a+b}.
\]
Applying this with $a=\mse_{\mathrm{priv}}$ and $b=\mse_{\mathrm{pub}}$, we obtain
\[
\min\{\mse_{\mathrm{priv}},\mse_{\mathrm{pub}}\}
\le
2\frac{\mse_{\mathrm{priv}}\cdot \mse_{\mathrm{pub}}}
{\mse_{\mathrm{priv}}+\mse_{\mathrm{pub}}}
=
2\mse_{\mathrm{aff}}.
\]
Combining the two inequalities yields
\[
\mse_{\thr} \le 2\mse_{\mathrm{aff}},
\]
and therefore
\[
\frac{\mse_{\thr}}{\mse_{\mathrm{aff}}}\le 2.
\]
This concludes the proof.

\section{Detailed Computation of Example}
Consider $\veps = (\frac{1}{2}, 1)$, The unique-threshold estimator has mean-squared error given by 
 \[
 \min \left(\frac{1}{4\cdot 2}+\frac{2}{(2\eps_1)^2}, \frac{1}{4}+\frac{2}{\eps_2^2} \right) = 17/8. 
 \]

On the other hand, for any $\tau>0$, $\mse_{\mathrm{aff}}(\veps)
 \le
 \frac{q_\tau}{4s_\tau^2}+\frac{2}{s_\tau^2}$, where
 $s_\tau=\sum_{i=1}^2 \min\{\eps_i,\tau\}$, and $q_\tau=\sum_{i=1}^2 \min\{\eps_i,\tau\}^2$. It is easy to check that $\tau=1$ is optimal, and we obtain
 \[
 s_1=\eps_1+\eps_2=\frac32,
 \qquad
 q_1=\eps_1^2+\eps_2^2=\frac14+1=\frac54.
 \]
 Hence
 \[
 \mse_{\mathrm{aff}}(\veps)
 =
 \frac{\frac54}{4(\frac32)^2}+\frac{2}{(\frac32)^2}
 =
 \frac{5}{36}+\frac{8}{9}
 =
 \frac{37}{36}.
 \]
 Combining the two bounds gives
 \[
 \frac{\mse_{\mathrm{th}}(\veps)}{\mse_{\mathrm{aff}}(\veps)}
 \ge
 \frac{\frac{17}{8}}{\frac{37}{36}}
 =
 \frac{153}{74}
 >2.
 \]

\section{Proof of Theorem~\ref{thm:two-level-factor-4}}\label{app:two-level}

Let $n = n_1 + n_2$. Because there are only two privacy levels, the Unique-threshold estimator has only two possible choices: either it uses all $n$ data points at privacy level $\eps_1$, or it uses only the $n_2$ data points in the second group at privacy level $\eps_2$. Hence its risk is
\[
\mse_{\mathrm{th}}(\veps)=\min\{\mse_1,\mse_2\},
\]
where
\[
\mse_1=\frac{1}{4n}+\frac{2}{(n\eps_1)^2},
\qquad
\mse_2=\frac{1}{4n_2}+\frac{2}{(n_2\eps_2)^2}.
\]

Next, define
\[
R:=1+\frac{8}{n_1\eps_1^2}.
\]
For two privacy levels, according to~\cite{Syomantak}, the optimal affine risk is given by
\[
\mse_{\mathrm{aff}}(\veps)
=
\begin{cases}
\displaystyle
\frac{n_1\eps_1^2+n_2\eps_2^2+8}{4(n_1\eps_1+n_2\eps_2)^2},
& \text{if } \eps_2\le R\eps_1,\\[1.2em]
\displaystyle
\frac{R}{4(n_1+n_2R)},
& \text{if } \eps_2\ge R\eps_1.
\end{cases}
\]
We show that in both regimes,
\[
\min\{\mse_1,\mse_2\}\le 4\,\mse_{\mathrm{aff}}(\veps).
\]

\paragraph{Case 1: $\eps_2\le R\eps_1$.}
In this case,
\[
\mse_{\mathrm{aff}}(\veps)
=
\frac{n_1\eps_1^2+n_2\eps_2^2+8}{4(n_1\eps_1+n_2\eps_2)^2}.
\]

We now compare this quantity to $\mse_1$ and $\mse_2$.

If
\[
n\eps_1 \ge n_2\eps_2,
\]
then
\[
n_1\eps_1+n_2\eps_2 \le 2n\eps_1.
\]
Moreover, since $\eps_2\ge \eps_1$,
\[
n_1\eps_1^2+n_2\eps_2^2+8 \ge n\eps_1^2+8.
\]
Therefore,
\[
\mse_{\mathrm{aff}}(\veps)
\ge
\frac{n\eps_1^2+8}{16n^2\eps_1^2}
=
\frac{1}{16n}+\frac{1}{2n^2\eps_1^2}
=
\frac{\mse_1}{4}.
\]
Hence
\[
\mse_{\mathrm{th}}(\veps)\le \mse_1 \le 4\,\mse_{\mathrm{aff}}(\veps).
\]

Otherwise,
\[
n_2\eps_2 \ge n\eps_1.
\]
Then
\[
n_1\eps_1+n_2\eps_2 \le 2n_2\eps_2,
\]
and
\[
n_1\eps_1^2+n_2\eps_2^2+8 \ge n_2\eps_2^2+8.
\]
Thus
\[
\mse_{\mathrm{aff}}(\veps)
\ge
\frac{n_2\eps_2^2+8}{16n_2^2\eps_2^2}
=
\frac{1}{16n_2}+\frac{1}{2n_2^2\eps_2^2}
=
\frac{\mse_2}{4}.
\]
Hence
\[
\mse_{\mathrm{th}}(\veps)\le \mse_2 \le 4\,\mse_{\mathrm{aff}}(\veps).
\]

So in all subcases of Case 1,
\[
\mse_{\mathrm{th}}(\veps)\le 4\,\mse_{\mathrm{aff}}(\veps).
\]

\paragraph{Case 2: $\eps_2\ge R\eps_1$.}
Set
\[
\eps_2':=R\eps_1.
\]
Since $\mse_2$ is decreasing in $\eps_2$, increasing $\eps_2$ can only help the threshold estimator. Therefore
\[
\mse_{\mathrm{th}}(\eps_1,\eps_2)\le \mse_{\mathrm{th}}(\eps_1,\eps_2').
\]

On the other hand, once $\eps_2\ge R\eps_1$, the affine optimum saturates: replacing $\eps_2$ by $\eps_2'=R\eps_1$ does not change $\mse_{\mathrm{aff}}$. Thus
\[
\mse_{\mathrm{aff}}(\eps_1,\eps_2)=\mse_{\mathrm{aff}}(\eps_1,\eps_2').
\]
Since $(\eps_1,\eps_2')$ lies on the boundary of Case 1, we may apply the previous argument to obtain
\[
\mse_{\mathrm{th}}(\eps_1,\eps_2')\le 4\,\mse_{\mathrm{aff}}(\eps_1,\eps_2').
\]
Combining the last three displays gives
\[
\mse_{\mathrm{th}}(\eps_1,\eps_2)\le 4\,\mse_{\mathrm{aff}}(\eps_1,\eps_2).
\]

Therefore, in all cases,
\[
\frac{\mse_{\mathrm{th}}(\veps)}{\mse_{\mathrm{aff}}(\veps)}\le 4.
\]
This concludes the proof.

\section{Proofs for the General Case of Section \ref{sec:general_case}}\label{app:general-case}

\subsection{Proof of Theorem~\ref{thm:lower-bound-general-m-levels}}

Assume that $2^m-1 \le n$. Let
\[
\eps_i = 2^{-(i-1)}
\qquad\text{and}\qquad
n_i = 2^{i-1}
\qquad\text{for all } i \in [m].
\]
Note that the total number of data points is given by
\[
\sum_{i=1}^m n_i = 2^m-1
\]
datapoints.

We first upper bound $\mse_{\mathrm{aff}}(\veps)$. Let
\[
S := \sum_{i=1}^m n_i \eps_i.
\]
For this construction,
\[
S = \sum_{i=1}^m 2^{i-1}2^{-(i-1)} = m.
\]
Consider the affine estimator assigning to each user in group $i$ the weight
\[
w_i=\frac{\eps_i}{m}.
\]
Then the weights sum to $1$, and
\[
\max_{i\in[m]}\frac{w_i}{\eps_i}=\frac1m.
\]
Moreover,
\[
\sum_{i=1}^m n_i w_i^2
=
\frac1{m^2}\sum_{i=1}^m 2^{i-1}2^{-2(i-1)}
<
\frac2{m^2}.
\]
Therefore,
\[
\mse_{\mathrm{aff}}(\veps)
\le
\frac14\sum_{i=1}^m n_i w_i^2
+
2\left(\frac1m\right)^2
\le
\frac{5}{2m^2}.
\]

We now lower bound $\mse_{\thr}(\eps)$. It suffices to consider thresholds
\[
\eps_{g+1}=2^{-g},
\qquad g\in\{0,\dots,m-1\}
\]
without loss of generality. For any threshold $\eps>0$, let
\[
n_\eps:=\bigl|\{j\in[n]: \eps_j\ge \eps\}\bigr|
\]
denote the number of selected datapoints. We have that
\[
n_\eps=2^{g+1}-1.
\]
Thus
\[
\eps n_\eps=(2^{g+1}-1)2^{-g}<2.
\]

Now take $\mathcal{P}=\delta_0$, so that $X_i=0$ almost surely and $\mup=0$. Then the threshold estimator reduces to
\[
\hat\mu_\eps(X)=Z_\eps,
\qquad
Z_\eps\sim \Lap\!\left(\frac1{n_\eps \eps}\right),
\]
and so
\[
\mathbb{E}\big[(\hat\mu_\eps(X)-\mup)^2\big]
=
\frac{2}{(n_\eps \eps)^2}
>
\frac12.
\]
Therefore,
\[
\mse_{\thr}(\eps)\ge \frac12.
\]

Combining the two bounds yields
\[
\frac{\mse_{\thr}(\veps)}{\mse_{\mathrm{aff}}(\veps)}
\ge
\frac{1/2}{5/(2m^2)}
=
\frac{m^2}{5}.
\]

\subsection{Proof of Theorem~\ref{thm:upper-bound-general-m-levels}}

We now prove a comparison lemma showing that for any candidate affine solution with clipping threshold $\tau$, one can choose a threshold whose error is worse by at most a factor $\min\{m^2,H_n^2\} \leq \min\{m^2,1 + \log n^2\} $, where
\[
H_n:=\sum_{k=1}^n \frac1k
\]
is the $n$-th harmonic number. Since the result holds for all $\tau$, it in particular holds for the optimal $\tau$, concluding the proof.

\begin{lemma}
\label{lem:threshold-vs-affine-sectionX}
For every $\tau>0$, there exists a threshold $\eps^*>0$ such that
\[
\mse_{\thr}
\le
\frac{1}{4n_{\eps^*}}+\frac{2}{(\eps^*)^2n_{\eps^*}^2}
\le
\min\{m^2,H_n^2\}
\left(
\frac{q_{\tau}}{4s_{\tau}^2}
+
\frac{2}{s_{\tau}^2}
\right),
\]
where
\[
n_{\eps^*}:=\bigl|\{j\in[n]:\eps{(j)}\ge \eps^*\}\bigr|.
\]
\end{lemma}

\begin{proof}
Fix $\tau>0$, and let
\[
\bar\eps_j:=\min\{\eps^{(j)},\tau\}
\qquad\text{for all }j\in[n].
\]
Reindex the coordinates so that
\[
\bar\eps^{(1)}\le \bar\eps^{(2)}\le \cdots \le \bar\eps^{(n)}.
\]
Since the original privacy vector takes at most $m$ distinct values, the clipped vector $(\bar\eps^{(j)})_{j=1}^n$ also takes at most $m$ distinct values.

For each $j\in[n]$, define
\[
u_j:=(n-j+1)\bar\eps^{(j)}.
\]
Let
\[
j^*\in\arg\max_{j\in[n]}u_j,
\qquad
\eps^*:=\bar\eps^{(j^*)}.
\]
Then
\[
u_{j^*}=(n-j^*+1)\eps^*.
\]

Moreover, all indices $j\ge j^*$ satisfy $\bar\eps^{(j)}\ge \eps^*$, hence also $\eps^{(j)}\ge \eps^*$. Therefore
\[
n_{\eps^*}\ge n-j^*+1,
\]
and so
\[
u_{j^*}\le \eps^* n_{\eps^*}.
\]

Next, by maximality of $j^*$, for every $j\in[n]$,
\[
u_j\le u_{j^*},
\qquad\text{that is,}\qquad
\bar\eps^{(j)}\le \frac{u_{j^*}}{n-j+1}.
\]
Summing over $j$ yields
\[
s_\tau=\sum_{j=1}^n \bar\eps^{(j)}
\le
u_{j^*}\sum_{j=1}^n \frac{1}{n-j+1}
=
u_{j^*}\sum_{k=1}^n \frac1k
=
H_n\,u_{j^*}.
\]
Thus
\[
s_\tau\le H_n\,u_{j^*}\le H_n\,\eps^* n_{\eps^*}.
\]

We also claim that
\[
s_\tau\le m\,u_{j^*}.
\]
Indeed, since $(\bar\eps_j)$ takes at most $m$ distinct values, the sum $s_\tau$ can be decomposed into at most $m$ blocks, each contributing at most $u_{j^*}$ by maximality. It must then be the case that $s_\tau \le m\,u_{j^*}$.

Combining the two estimates, we obtain
\[
s_\tau\le \min\{m,H_n\}\,u_{j^*}
\le
\min\{m,H_n\}\,\eps^* n_{\eps^*}.
\]
It follows that
\[
\frac{2}{(\eps^*)^2n_{\eps^*}^2}
\le
\min\{m^2,H_n^2\}\frac{2}{s_\tau^2}.
\]

For the sampling term, every selected point satisfies $\bar\eps^{(j)}\ge \eps^*$, so
\[
q_\tau=\sum_{j=1}^n (\bar\eps^{(j)})^2
\ge
n_{\eps^*}(\eps^*)^2.
\]
Combining this with
\[
s_\tau\le \min\{m,H_n\}\,\eps^* n_{\eps^*}
\]
gives
\[
\frac{q_\tau}{4s_\tau^2}
\ge
\frac{n_{\eps^*}(\eps^*)^2}{4s_\tau^2}
\ge
\frac{1}{4\min\{m^2,H_n^2\}n_{\eps^*}},
\]
that is,
\[
\frac{1}{4n_{\eps^*}}
\le
\min\{m^2,H_n^2\}\frac{q_\tau}{4s_\tau^2}.
\]

Summing the two bounds yields
\[
\frac{1}{4n_{\eps^*}}+\frac{2}{(\eps^*)^2n_{\eps^*}^2}
\le
\min\{m^2,H_n^2\}
\left(
\frac{q_{\tau}}{4s_{\tau}^2}
+
\frac{2}{s_{\tau}^2}
\right).
\]
Since $\mse_{\thr}$ is the infimum over all thresholds, it is at most the left-hand side for this choice of $\eps^*$, which proves the claim.
\end{proof}